\documentclass[11pt,letter,table,usenames,dvipsnames]{article}
\usepackage{verbatim} 

\usepackage{setspace}
\usepackage{amssymb}
\usepackage{amsmath}
\usepackage{amsthm} 
\usepackage{graphicx}
\usepackage{hyperref}
\usepackage{mathpazo}

\hypersetup{colorlinks=true,citecolor=Maroon,urlcolor=Maroon, linkcolor=Maroon,anchorcolor =Maroon}
\usepackage{natbib}
\usepackage{lscape}
\usepackage{tikz}
\newcommand*\circled[1]{\tikz[baseline=(char.base)]{
            \node[shape=circle,draw,inner sep=2pt] (char) {#1};}}

\usepackage[left=1in, right=1in, top=1in, bottom=1in]{geometry}

\usepackage{aliascnt}
\usepackage[nameinlink]{cleveref}
\crefname{equation}{equation}{equations}
\creflabelformat{equation}{#2#1#3}

\newtheorem{defi}{Definition}

\newtheorem{proposition}{Proposition}

\newaliascnt{lemma}{theorem}
\newtheorem{lemma}[lemma]{Lemma}
\aliascntresetthe{lemma}
\crefname{lemma}{Lemma}{Lemmas}
\Crefname{lemma}{Lemma}{Lemmas}

\begin{document}

\title{\textbf{A Linear Model of Geopolitics}\footnote{Li: University of Massachusetts; Zhang: Tsinghua University. The paper was completed when both authors were affiliated with Boston College. We thank the editor and two anonymous reviewers for their helpful comments. We also thank Jim Anderson, Leonardo Baccini, Susanto Basu, Emily Blanchard, Richard Chisik, Kim-Sau Chung, Dave Donaldson, Thibault Fally, Pedro Gomis-Porqueras, Wen-Tai Hsu, Wolfgang Keller, Tilman Klumpp, Hideo Konishi, Arthur Lewbel, Kalina Manova, Thiery Mayer, Steve Redding, John Ries, Dan Trefler, and participants at various seminars and conferences for their comments on the previous draft. The standard disclaimer applies.}}
\author{Ben G. Li  \and Penglong Zhang }
\date{This draft: \today}

\maketitle
\thispagestyle{empty}

\begin{abstract}
Geopolitics is shaped by trade and borders. We develop a general-equilibrium model in which both are endogenously determined in a linear world. Their interaction rationalizes geopolitical outcomes that cannot be obtained when either trade or borders are treated as exogenous. This unified and tractable framework is used to study political economy, security, and ideology within and across states. \textit{(JEL Codes: F50, P16, N40)}


\end{abstract}

\section{Introduction}

Geopolitics arises from the interaction between trade and borders. Borders restrict trade, while trade reshapes the value and stability of borders. When external trade becomes more important than internal trade, existing borders no longer reflect underlying economic interests and come under pressure to adjust. Treating either trade or borders as fixed therefore provides only a partial-equilibrium view. A tractable framework that endogenizes both has long been missing from the economist’s toolkit.

Borders are difficult to model because they are consensual in formation and multilateral in effect. A polity cannot unilaterally place a border within another polity’s territory; borders require mutual consent and are shared between adjacent polities. Once neighboring polities jointly establish borders, they generate political and economic externalities for non-participating ones. Border formation thus appears bilateral but is inherently multilateral. Establishing equilibrium with heterogeneous polities is already nontrivial, even without modeling trade—which likewise appears bilateral but is inherently multilateral.

In this paper, we present a general equilibrium model with endogenous trade and endogenous borders. We consider a continuum of locales located along a world line. Domestic trade is less costly than foreign trade, encouraging locales to form larger states, whereas larger states face higher governance costs, pushing locales to form smaller states. The world line therefore operates as a market for statehood, in which locales select others with whom to form states. This model yields borders with analytical characterizations, generating tractable geopolitical interactions both within and across states.

Locales that are geographically proximate to others enjoy locational advantages in trade. While the trade literature has long emphasized such trade-cost advantages, it has largely abstracted from their geopolitical implications. In our model, locales join neighboring states that are more proximate to the rest of the world. The resulting state system, represented as a partition of the world line, admits a unique equilibrium (\Cref{sec:theory}). The model provides a unified framework for analyzing a broad class of contemporary geopolitical questions. It addresses political economy issues such as trade and migration (\Cref{sec:polecon}); extends to political security, including border stability, state size, and suffrage expansion (\Cref{sec:polgeo}); and applies to political ideology, including national opinion and regional separatism (\Cref{sec:polideo}).


We contribute to two bodies of literature. The first examines states and borders, which have largely evolved as two distinct branches. The “state” branch theorizes optimal state size by balancing economies of scale against the costs of governing large states \citep*{Friedman77,BB80,AS97,ASW00,AS05,ASW05,AS06,DLOW11,GPV16}. In these models, borders are either absent or introduced solely to compute state size. As a result, borders play no material role, limiting the scope of this literature for geopolitical analysis. The “border” branch studies the formation of state borders using simulation \citep*{TCTG13,Weese16,Allen23,FKLS23}. Because borders typically lack regular shapes—and interactions among multiple borders further compound this irregularity—endogenizing borders quickly encounters the curse of dimensionality. Existing approaches therefore rely on fixed real-world maps or specific algorithms to discipline border formation. While these methods address borders directly, deriving general geopolitical implications remains difficult as outcomes depend on the chosen map or algorithm and must be obtained through simulation.

We use a linear world to integrate state size and state border decisions in a single model. Economics has a long tradition of employing one-dimensional environments to isolate mechanisms that are otherwise obscured, including \citet{Hotelling29} on spatial competition, \citet{Black48} and \citet{Downs57} on majority-rule voting, \citet*{DFS77} on comparative advantage, and \citet{OF80} on urban structure. In a linear world, borders are points on a line and therefore admit analytical characterization. State size and state borders become two sides of the same object: solving for one immediately determines the other. This framework allows us to model geopolitics through border behaviors. As noted earlier, consolidating trade and borders is challenging because both involve multilateral interests and are jointly determined. In our linear setting, these interdependencies remain fully endogenous yet become tractable. Linearization transforms an otherwise intractable high-dimensional general-equilibrium problem into a feasible border-drawing problem.


The second body of literature studies the interaction between trade and politics. A large body of work examines how trade shapes domestic politics, including checks and balances \citep*{AJR05}, parliamentary behavior \citep{PT14}, democratization \citep{GT14}, military conflict \citep*{SS01,MMT08b,AY10,BO14}, contract enforcement \citep*{AM02,RL07,Anderson09}, and social organization \citep*{Greif94}. Another extensive literature investigates how trade interacts with international politics, particularly various international relations \citep*{BB02,Krishna03,BB04,ELSW11,KS14,GHM16,GT21,CMS25a,CMS25b}. All modern politics, domestic and international, rests on states as their fundamental units. By treating the state itself as an outcome of trade interests held by polities, our framework integrates domestic and international politics through the endogenous formation, evolution, and dissolution of states.


The rest of the paper is structured as follows. In \Cref{sec:theory}, we present the model and establish the existence and uniqueness of equilibrium. In \Cref{sec:polecon,sec:polgeo,sec:polideo}, we apply the model to political economy, security, and ideology, respectively, within and across states. In \Cref{sec:conclude}, we conclude.

\section{The Baseline Model \label{sec:theory}}

\subsection{Environment}

Consider a world represented by a continuum of locales, indexed by \( t \in [-1,1] \). The midpoint locale \( t = 0 \), serving as the World Geometric Center (WGC), divides the world into left and right hemispheres. In each hemisphere, there are two directions: proximal (towards the WGC) and distal (away from the WGC). For expositional convenience, our analysis often uses the right hemisphere, while the same arguments apply symmetrically to the left hemisphere.

All locales have identical quantities of land \( z \) and initial labor \( l^0 \), both inelastically supplied to produce locale-specific differentiated goods. Locales have the same production function
\begin{equation}
y(t) = z(t)^{\alpha} l(t)^{1-\alpha}, \label{eq:prodfn}
\end{equation} 
where \( 0 < \alpha < 1 \), \( z(t) \) represents the land at locale \( t \), and \( l(t) \) the labor at locale \( t \). The land \( z(t) \) is immobile (i.e., affixed to locale \( t \)) and owned by the lord of locale \( t \). Labor  within a state (defined later) migrates freely. Firms compete perfectly in production and sales.

At each locale $t$, the lord and labor have an aggregate consumption of $C(t)=C^z(t)+C^l(t)$. The lord has utility function
\begin{equation}
U(t) = \frac{1}{1-\gamma} C^z(t)^{1-\gamma} - h S(t),
\label{eq:utilityfn}
\end{equation}
where $\gamma > 1$, $h > 0$, and $h S(t)$ represents governance costs proportional to state size $S(t)$. The governance costs arise from congestion: as states grow larger, domestic conflicts intensify, increasing the coordination burden borne by lords \citep{ASW00}. The labor has utility function 
\begin{equation}
V(t)=\frac{\psi}{1-\gamma} C^l(t)^{1-\gamma},\label{eq:utilityfn0}
\end{equation}%
where $\psi>0$ is a free scalar that allows a potential difference in their marginal utility of consumption. The consumption of both the lord and the labor consists of all goods produced globally, aggregated using a Cobb-Douglas function:
\begin{equation}
C(t)\equiv \exp \left(\int_{-1}^1 \ln c(t,s) ds\right),  \label{eq:CES}
\end{equation}%
where $c(t,s)$ is the quantity of the good made by locale $s$ and consumed at locale $t$.

Without loss of generality, we assume that consumers bear the trade costs. Specifically, to consume \( c(t,s) \) units of the good produced at locale \( s \), consumers at locale \( t \) pay for \( c(t,s) \) and a trade cost. We specify the trade cost as
\begin{equation}
d(t,s)= \exp(a(t,s) g(t,s)),
\end{equation}
where $a(t,s)$ represents a political cost and $g(t,s)$ a physical cost. As in their analysis, we assume $a(t,s)=0$ for domestic trade and $a(t,s)=1$ for interstate trade. Unlike their symmetric locales, our locales differ in locations such that $g(t,s)= \tau |s-t|$, where $\tau>0$ is the physical cost per unit of distance.\footnote{To introduce physical trade costs that vary with distance, we follow the trade literature to use the exponential function form (see \citet{AVW04} for a review), micro-founded by the aggregation of incremental iceberg costs as the distance between the increments tends to zero \citep{AA14}.} Thus,
\begin{equation}
d(t,s)=\left\{ 
\begin{array}{ll}
1, & \text{if }s\in n_t, \\ 
\inf_{\tilde{t}\in n_t}\exp (\tau \left\vert s-\tilde{t}\right\vert ), & \text{if }%
s\not\in n_t,%
\end{array}%
\right.  \label{eq:ddef}
\end{equation}
where $n_t$ refers to the state to which locale $t$ belongs. The limit inferior $\inf_{\tilde{t}\in n_t}$ indicates that trade costs apply only beyond the farthest national border. $d(t,s)$ is in the iceberg form: only one unit of the good reaches locale \( t \) if \( d(t,s) \ge 1 \) units are shipped from locale \( s \) to locale \( t \). 

The world is partitioned into states by a set of borders:
\begin{equation}
\{b_{n}\} \equiv 
\{b_{-N},...,b_{-1},b_{-0},b_{0},b_{1},...,b_{N}\},\label{eq:partitiondef}
\end{equation}
where $b_{-0} \le 0 \le b_{0}$ and $-1 \le b_{-N}< b_N \le 1$. Hereafter, State $n>0$ ($n<0$) denotes a state in the right (left) hemisphere with size $S_n>0$:
\begin{equation}
S_n= b_{n} - b_{n-1}.
\end{equation}
State $0$, referring to the state formed by $b_{-0}$ and $b_{0}$, has size $S_0=b_{0}-b_{-0}$ and is in both hemispheres. Note that in both hemispheres, the left border of state $n$ is $b_{n-1}$, the right border $b_{n}$. As a convention, we let the distal (proximal) border of every state be an open (closed) endpoint. For example, if state $n$ is in the right hemisphere, its left border $b_{n-1}$ is part of state $n$ while its right border $b_{n}$ is not part of state $n$ but part of state $n+1$. Since locales are located on a continuum, we require that locales form states only with their adjacent locales (i.e., enclaves are not allowed).

Now we introduce  
\begin{defi}
$\{b_n^*\}$ is an equilibrium world partition if it satisfies criteria (i) and (ii):
\begin{equation}
\text{(i) }b^{L}(t)=\sup \{b_n^* | b_n^* < t\} \text{ and } b^{R}(t)=\inf \{b_n^* | b_n^* > t\}
\text{ for any } t \in [-1,1], \label{eq:equi_line1} 
\end{equation}
and
\begin{multline}
\text{(ii) For any $\widehat{b}^{L}(t) \neq  b^{L}(t)$ and $\widehat{b}^{R}(t) \neq  b^{R}(t)$, if $U(t | \widehat{b}^{L}(t),\widehat{b}^{R}(t))> U(t| {b}^{L}(t),{b}^{R}(t))$}, \\ \text{there must be at least one $t' \neq t $ such that $U(t' | \widehat{b}^{L}(t),\widehat{b}^{R}(t))< U(t'| {b}^{L}(t),{b}^{R}(t))$}.\label{eq:nobetter}
\end{multline}
\end{defi} 
\noindent In the definition, criterion (i) links locales with their states through geometric coordinates. Locales in the same state share left and right borders. Criterion (ii) ensures that no lord can benefit from deviating from the equilibrium partition without harming other lords. Since labor does not participate in the drawing of borders, we hereafter use the terms \textit{locale} and \textit{lord} interchangeably in border-drawing affairs. Since borders imply states, the terms \textit{equilibrium world partition} (i.e., $\{b_n^*\}$) and \textit{equilibrium states} are equivalent. If excluding a lord from a state can improve the welfare of other lords in that state, then the state with the included lord is not part of the equilibrium partition. Moreover, even if all lords in a state agree on the state's borders, the state is still not part of the equilibrium partition if there exists one foreign lord that wants to join the state and allowing it to join does not harm any existing lord. 

To close the model, we need the markets of goods to be in equilibrium. Given that all locales have identical factor endowments and the Cobb-Douglas production and consumption structures, we set the factory-gate price of all goods to be \( p \), regardless of their origins. Firms at locale \( s \) are indifferent to the destination of their sales. The market clearing condition for the good produced at locale \( s \) is given by
\begin{equation}
\int_{-1}^{1}y(t,s)dt=y(s), \label{eq:marketclear}
\end{equation}
where \( y(s) \) is the total output of locale \( s \). The market-clearing condition \eqref{eq:marketclear} is invariant across origins (e.g., between $s$ and $s'$), ensured by the Cobb-Douglas consumption structure \eqref{eq:CES} and the fact that consumers pay trade costs. 

Now we define an equilibrium of the model:  
\begin{defi}
An equilibrium of the model takes the form of 
\begin{equation}
\Omega \equiv \left(b^{L}(t), b^{R}(t), C^z(t), C^l(t), l(t), y(t), \forall t \in [-1,1] \right), \label{eq:equilibrium_looking}
\end{equation}
where
\begin{equation}
\{b^{L}(t), b^{R}(t),\forall t\}=\{b_{n}^*\}. 
\end{equation}
\end{defi}

The $\Omega$ defined above constitutes a subgame-perfect equilibrium, since borders are chosen prior to production and consumption. To translate border choices into states, we introduce overlords, who draw borders on behalf of fellow lords. The timing of events is therefore as follows.
\begin{description}

\item[{Date 1}] The lord of locale \( t \) makes offers to the lords of his proximal locales and himself to be his overlord. If any of these offers is accepted by a proximal lord, that lord becomes his overlord. If multiple offers are accepted by proximal lords, the most proximal one among them becomes his overlord. If none of the offers are accepted by proximal lords, the lord becomes his own overlord and decides whether to accept any offer from his distal locales. All lords act simultaneously to make offers, then decide on offers, and settle their roles as overlords or (ordinary) lords. For the lord of an arbitrary locale $t$, his overlord is the lord at a locale $t^*$ that satisfies $|t^*| \le |t|$.

\item[{Date 2}] Overlords set borders for themselves and distal neighbors whose offers they have taken. The set of overlords $\{t^*\}$ informs the set of borders \(\{b_n^*\}\), partitioning the world into states.

\item[{Date 3}] Production and consumption occur. Lords (including overlords) and labor forces produce goods, which are then traded and consumed within and across borders.
\end{description}
The three dates occur sequentially, with each providing information that is remembered by agents on subsequent dates.


\subsection{Solving the equilibrium}
The labor at locale $t$ solves
\begin{align}
\max_{\{c^l(t,s)\}_{s},\, l(t)} \; V(t) \quad \text{s.t.} \quad \int_{-1}^{1} p(t,s)c^l(t,s)ds = w(t)l(t). \label{eq:laboroptimization}
\end{align}
Since labor does not participate in border drawing, all labor decisions are made on Date 3. The lord at locale $t$ solves
\begin{align}
\label{eq:lordoptimization}
\max_{\{c^z(t,s)\}_{s},\, z(t),\, b^L(t),\, b^R(t)} \; U(t) \quad \text{s.t.} \quad  
\begin{aligned}[t]
\int_{-1}^{1} p(t,s)\, c^z(t,s)\, ds &= r(t)z(t), \\
|b^L(t)-b^R(t)| &= S_{n_t}, \\
b^L(t) \le t & \le b^R(t).
\end{aligned}
\end{align}
Like labor, the lord’s factor input and consumption decisions are made on Date 3. Unlike labor, the lord also chooses borders on Date 2. Among the lords, those who are not overlords delegate their border-drawing decisions to their overlords on Date 1. Below, we solve the model by backward induction.

\paragraph{Solving Date 3.} The total expenditure at locale $t$ on the good produced by locale $s$ equals\footnote{See Appendix \ref{aggCproof} for the derivation.}
\begin{equation}
p(t,s)c(t,s)=\frac{C^z(t)^{1-\gamma}}{\lambda^z(t)}+\frac{\psi C^l(t)^{1-\gamma}}{\lambda^l(t)}  \equiv \kappa(t), \label{eq:aggC}
\end{equation}
where $\lambda^z(t)$ and $\lambda^l(t)$ are the shadow prices. By taking the integral of equation \eqref{eq:aggC} across destination locale $t$, we obtain the GDP of the good's origin locale $s$:
\begin{equation}
py(s)=\int_{-1}^1 p(t,s)c(t,s)dt=\int_{-1}^1 \kappa(t) dt \equiv \kappa. \label{eq:localegdp}
\end{equation}
Since term $\int_{-1}^1 \kappa(t) dt$ in equation \eqref{eq:localegdp} does not vary by $s$, $y(s)$ is locale-invariant as well. Intuitively, trade costs are all paid by consumers and thus the Cobb-Douglas consumption structure ensures that all locales face the same global demand. Thus, we denote locale GDP uniformly by $\kappa$ in equation \eqref{eq:localegdp}. Then, factor incomes for lord and labor are given by
\begin{align}
r(s) z(s) &= \alpha p y(s) = \alpha \kappa, \label{eq:factorincomeZ}\\
w(s) l(s) &= (1-\alpha) p y(s) = (1-\alpha)\kappa, \label{eq:factorincomeL}
\end{align}
for every locale ($s$) in the world.

On Date 3, every state has a statewide labor market. In this market, the total labor supply is the aggregate of the initial labor $l^0$ across locales of the state, while the total labor demand is the aggregate of the $l(s)$ in equation \eqref{eq:factorincomeL} across locales of the state. Since land is immobile within a state, the resulting wage rate is equalized across locales within the state. That is, for any state $n$, its initial labor will be distributed uniformly across locales in equilibrium: 
\begin{equation} 
l(s)=l(s') \text{ for any } s,s' \in n,\text{       and       } \int_{s \in n } l(s)ds =\int_{ s \in n } l^0ds.\label{eq:equilabor}
\end{equation}
Then we have
\begin{equation} 
y= z^{\alpha}{l^0}^{1-\alpha}, \text{       and       } py=rz+w l^0, \label{eq:allprices}
\end{equation} 
for every locale in the world.

\sloppy Three observations are immediate. First, the locale-level variables defined in equations
\eqref{eq:factorincomeZ}--\eqref{eq:allprices} are equalized across locales and may therefore be represented by the common vector $\{p, y, r, z, w, l\}$. Second, these variables are invariant to the world partition, implying perfect foresight over them for all agents prior to Date~3. Third, the remaining equilibrium objects in $\Omega$ are $C^z(t)$, $C^l(t)$, $b^{L}(t)$, and $b^{R}(t)$, all of which depend on the world partition $\{b_n^*\}$ determined on Date~2.

\paragraph{Preparation for Solving Date 2.} 
Among the locale-invariant prices $p$, $r$, and $w$, one can be normalized. Normalizing $p = \frac{z}{2} r$, we obtain
\begin{lemma} \label{lemma:suff}
\begin{equation}
C^z(t)=1/R(t),\label{eq:suff}
\end{equation}
where 
\begin{equation}
R(t)\equiv \exp \left(\int_{-1}^1 \ln d(t,s) ds\right). \label{eq:remoteness0}
\end{equation} 
\end{lemma}
\begin{proof}
See Appendix \ref{suffproof}.  
\end{proof}

The $R(t)$ in \Cref{lemma:suff}, defined as an aggregate of locale $t$'s bilateral distance from the rest of the world, measures locale $t$'s \textit{remoteness} from the rest of the world. It can alternatively be interpreted as the price index faced by locale $t$'s consumers. Since all lords in the world have the same income, their real income can be expressed as $1/R(t)$ after the normalization. The function $R(t)$ greatly facilitates the analysis of  the world partition  $\{b_{n}^*\}$ due to its following properties.

First, $R(t)$ applies to all locales in the same state as locale $t$, because domestic trade cost is zero. Denote the state of locale $t$ by $n_t$, then 
\begin{align}
R(t)&=  \exp \left( \int_{ - 1}^{{b_{n_{t - 1}}}} {\tau ({b_{n_t - 1}} - s)ds}    + \int_{{b_{n_t}}}^1 {\tau (s - {b_{n_t}})ds} \right) \\
&=\exp \left( \frac{\tau}{2} [ (1+b_{n_{t-1}})^2 + (1-b_{n_t})^2 ]\right) \equiv R_{n_t},  \label{eq:remoteness}
\end{align}%
where the first (second) term in the exponential function corresponds to the remoteness to the left-side (right-side) world. The $n_{t-1}$ in equation \eqref{eq:remoteness} refers to the left-side neighboring state of state $n_t$. The right border of state $n_{t-1}$, namely $b_{n_{t-1}}$, is the left border of state  $n_{t}$. By equation \eqref{eq:remoteness}, $R(t)$ is increasing in state $n_t$'s minimal distance from the WGC.\footnote{If locale $t$ is in the right hemisphere (i.e., $t>0$), $b_{n_{t-1}}$ is the point where locale $t$ and its fellow locales start paying trade costs for their imported goods from their left-side foreign states. If locale $t$ is in the left hemisphere (i.e., $t<0$), $b_{n_t}$ is the point where locale $t$ and its fellow locales start paying trade costs for their imported goods from their right-side foreign states.} 

Now consider the right hemisphere of the world, where state $n$ is the $n$-th nearest state to the WGC. Its left border is $b_{n-1}$, its right border $b_{n}$, its size $S_n = b_n-b_{n-1}$, and its remoteness $R_n$. The following properties of  $R_n$ follow from equation \eqref{eq:remoteness}:
\begin{align}
\frac{{\partial {R_n}}}{{\partial {S_n}}} &=  - \tau (1 - {b_{n - 1}} - {S_n}){R_n}<0, \label{eq:RSeq}\\
\frac{{\partial {R_n}}}{{\partial {b_{n - 1}}}} &= \tau (2{b_{n - 1}} + {S_n}){R_n}>0,\label{eq:Rleftborder}\\
\frac{{\partial {R_n}}}{{\partial \tau }} &= \frac{1}{2}[{(1 + {b_{n - 1}})^2} + {(1 - {b_{n - 1}} - {S_n})^2}]{R_n}>0.\label{eq:Rtau}
\end{align}
According to equation \eqref{eq:RSeq}, $R_n$ decreases if state $n$ increases in size. The state size change may take the form of (a) fixing the left border and pushing the right border rightward, (b) fixing the right border and pushing the left border leftward, or (c) pushing both borders outward. According to equation \eqref{eq:Rleftborder}, $R_n$ decreases if state $n$ moves leftward with its size unchanged. According to equation \eqref{eq:Rtau}, $R_n$ decreases if no border changes but the foreign trade cost per unit of distance ($\tau$) decreases. These properties of $R(t)$ prepare us to solve for the world partition $\{b_n^*\}$ on Date 2.

\paragraph{Solving Date 2.}  On Date 2, the overlords (designated on Date 1, as discussed later) draw borders for themselves and other lords in their states, with perfect foresight of the events on Date 3 (as discussed above). For the lord of locale \( t \), a marginally larger state results in a disutility \( h \) (see utility function \eqref{eq:utilityfn}) and a consumption gain (see \Cref{lemma:suff} and property \eqref{eq:RSeq}).

To solve for the world partition \(\{b_{n}^*\}\), we start from the center of the world and move outward towards its two ends. This procedure yields an equilibrium of the model, which we later prove to be unique. Consider locale \( t=0 \) located at the WGC. It is the most proximal locale in the world, and therefore its lord must be an overlord.\footnote{The lord at the WGC can only take his own offer to be an overlord.} By equation \eqref{eq:remoteness}, it has the lowest possible remoteness. Thus, if locale \( t=0 \) draws its borders to include any other locale in the world as its fellow locale, that locale will attain the lowest possible remoteness, and its utility cannot increase further. To determine the first locale in each hemisphere that is not included, the lord at locale $t=0$ solves problem \eqref{eq:lordoptimization} by imposing
\begin{equation}
b^{L}(t=0)=-b^{R}(t=0)=\frac{S(t=0)}{2}.
\end{equation}
The resulting first-order condition is
\begin{equation}
\tau R_0^{\gamma - 1}(1 - (b_0^* - b_{-0}^*)) = h,
\label{eq:FOC00}
\end{equation}
where $b_{-0}^*$ and $b_0^*$ denote the optimal borders for state 0. Since all locales within state 0 share the same remoteness, $R(t=0)$ applies to the entire state and is denoted by $R_0$. Likewise, $S(t=0)$ applies to the entire state and is denoted by $S_0$.

The locale at $b_0^*$ is excluded from state 0 and its lord therefore forms a state to the right of state 0. This lord solves problem \eqref{eq:lordoptimization} by choosing borders such that
\begin{equation}
b^{L}(t = b_0^*) = b_0^*, \quad b^{R}(t = b_0^*) = b_0^* + S(t = b_0^*).
\end{equation}
This new state, denoted state 1, extends rightward until the following first-order condition holds for the lord:
\begin{equation}
\tau R_1^{\gamma - 1}\bigl(1 - b_0^* - S_1\bigr) = h,
\label{eq:FOC11}
\end{equation}
which pins down state size \( S_1 \) and hence the right border \( b_1^* = b_0^* + S_1 \). The locale at \( b_1^* \) is excluded from state 1. All locales in the interval \( [b_0^*, b_1^*) \) achieve their minimal remoteness by joining state 1, and their utility cannot be further increased. These locales therefore share \( R_1 = R(t=b_0^*) \) and \( S_1 = S(t=b_0^*) \). 

State 1 is the first state in the right hemisphere. For state \( n \ge 1 \), the locale at \( b_{n-1}^* \) is excluded from state \( n-1 \), and its lord forms a state to the right of state \( n-1 \). The lord solves problem~\eqref{eq:lordoptimization} by choosing borders
\begin{equation}
b^{L}(t = b_{n-1}^*) = b_{n-1}^*, \qquad 
b^{R}(t = b_{n-1}^*) = b_{n-1}^* + S(t = b_{n-1}^*).
\end{equation}
The resulting first-order condition is
\begin{equation}
\tau R_n^{\gamma - 1}\bigl(1 - b_{n-1}^* - S_n\bigr) = h,
\label{eq:FOC}
\end{equation}
which pins down state size \( S_n \) for state \( n \) and the right border \( b_n^* = b_{n-1}^* + S_n \). The locale at \( b_n^* \) is excluded from state~\( n \). All locales in the interval \( [b_{n-1}^*, b_n^*) \) achieve their minimal remoteness by joining state \( n \), and therefore share the same values \( R_n \) and \( S_n \) as the locale at \( b_{n-1}^* \). This solving process applies symmetrically to the left hemisphere. 

In equilibrium, all borders in the world (i.e., the partition $\{b_n^*\}$) are determined. The total number of states is $2N+1$, where $2N$ satisfies\footnote{When state 0 collapses to a stateless point, namely \(b_0^* = b_{-0}^* = 0\), the number of states is $2N$ rather than $2N+1$. A necessary condition for this extreme scenario to occur is \(\tau \exp(\tau (\gamma-1)) = h\) (see first-order equation \eqref{eq:FOC00}). There is no particular reason for this parameter combination to arise, and we therefore exclude this special case from consideration. Nonetheless, in this special case, the first-order condition \eqref{eq:FOC} still holds, as does the rest of the equilibrium.}
\begin{equation}
2N
= \left\{
2n :
\frac{S_0}{2} + \sum_{i=1}^{n} S_i \le 1
\;\text{and}\;
\frac{S_0}{2} + \sum_{i=1}^{n+1} S_i > 1
\right\}.
\label{eq:stateno}
\end{equation}
This state system leaves very distal locales outside any state---locales in the intervals
$[-1, b_{-N}]$ and $[b_N, 1]$ are not incorporated into their proximal states. At each pole, the remaining locales have incentives to form their own states in order to eliminate foreign trade costs among themselves. These two polar states, unlike the $2N+1$ interior states, are smaller than the optimal size; we therefore refer to them as \textit{polar semi-states}. They serve to maintain a stable number of states in the linear world.

Thus far, a world partition $\{b_n^*\}$ has been solved for. One might wonder whether other world partitions could emerge if borders were determined in a different order, a question we address below.

\paragraph{Solving Date 1.} The world partition $\{b_n^*\}$ determined on Date~2 consists of a set of critical coordinates: locales at these coordinates become the most proximal locales within their respective states. The lords at these locales act as overlords, drawing borders for themselves and other lords in their states. These overlords are designated on Date~1 by lords who have perfect foresight regarding Date~2 outcomes.

Belonging to $\{b_n^*\}$ is both necessary and sufficient for being an overlord. Sufficiency follows because none of the lords at these locales can identify a more proximal locale to serve as an overlord, as shown above, and therefore must serve as their own overlords. Necessity follows from the fact that no overlords exist outside the set of locales $\{b_n^*\}$, a claim we explain below and prove later as part of a proposition.

First and foremost, the composition of state~0 is determined by utility maximization. Consequently, regardless of the starting point used to solve for borders, the overlord at locale $t=0$ selects the state $(b_{-0}^*, b_0^*)$. All locales within this interval find state~0 optimal, since belonging to the same state as locale $t=0$ minimizes their remoteness and attains the global minimum of remoteness. As a result, no other overlord exists within $(b_{-0}^*, b_0^*)$.

All other states formed by the partition $\{b_n^*\}$ likewise achieve the lowest possible remoteness for their constituent locales, given the locations of those locales. Equivalently, any alternative partition of the world, denoted by $\{\tilde b_k\}$, would generate strictly higher remoteness for at least some locales relative to the allocation induced by $\{b_n^*\}$. To see why, consider the right hemisphere and suppose there exists some $k \ge 1$ such that $\tilde b_n = b_n^*$ for $n = 0,1,\ldots,k-1$, but $\tilde b_k \neq b_k^*$.\footnote{The same argument applies symmetrically to the left hemisphere.} \Cref{fig:uniqueness} illustrates why such a configuration for state $k$ cannot constitute an equilibrium partition. If $\tilde b_k < b_k^*$, corresponding to the middle panel, locales in the $\bigstar$ region are strictly worse off than under the equilibrium partition, since their remoteness would depend on $\tilde b_k$ rather than on $b_{k-1}^*$. Moreover, the lord at locale $\tilde b_k$ would prefer to join state $k$ in the upper panel—where the most proximal locale is $b_{k-1}^*$—rather than to rule over state $k+1$, where the most proximal locale would be his own.\footnote{Even if the lord at $\tilde b_k$ were to accept offers from locales in the $\bigstar$ region, such acceptance would not make him an overlord, since those locales would also make offers to $t = b_{k-1}^*$, who would accept them into his state. Recall that when multiple accepted offers exist, only the most proximal one is effective.} Consequently, no overlord can arise between $b_{k-1}^*$ and $b_k^*$. Conversely, if $\tilde b_k > b_k^*$, as illustrated in the lower panel, the resulting state would be too large for the overlord at $\tilde b_{k-1}$, who determines $\tilde b_k$ for state $k$. The optimal size of state $k$ is pinned down by the first-order condition \eqref{eq:FOC}, which would be violated in this case. Therefore, in either scenario, the proposed state $k$ cannot be part of an equilibrium partition.

\begin{figure}[h!]
\centering
\caption{State $k$ is Not Part of Any Other Equilibrium Partition}
\label{fig:uniqueness}
\hspace*{0cm}\includegraphics[scale=.8,clip]{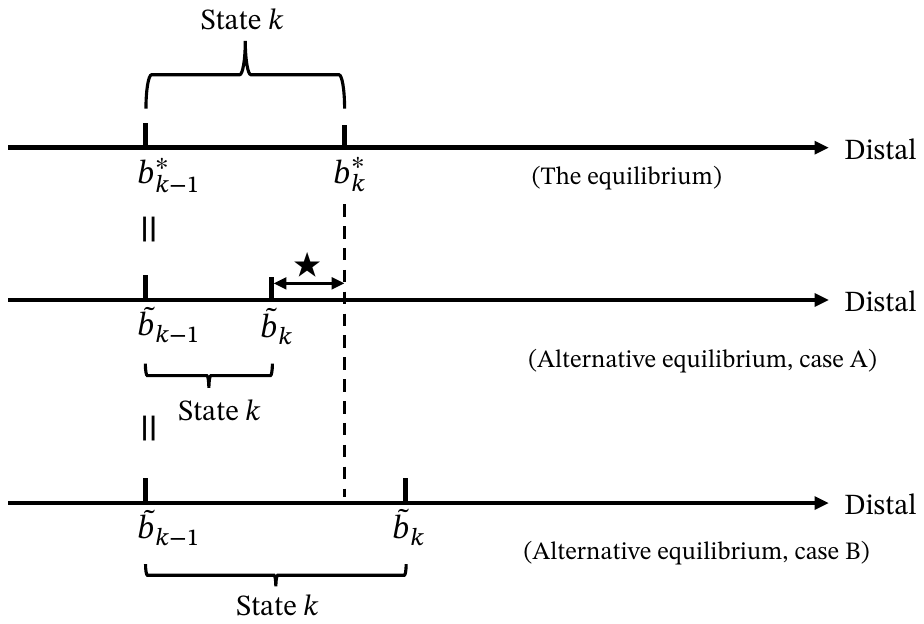}
\vspace{-400pt}
{ \par \raggedright \footnotesize \textit{Notes:}  The figure is concerned with the right hemisphere. The state $k$ in the upper panel, formed by borders $b_{k-1}^*$ and $b_k^*$, is a state in the equilibrium partition we described in the text. Suppose that another equilibrium partition exists. In this alternative equilibrium partition, state $k$ has the same left border $\tilde{b}_{k-1}=b_{k-1}^*$ but a different right border $\tilde{b}_k$ than $b_k^*$. (The smallest possible $k$ is 1; that is, $\tilde{b}_{k-1}=b_{k-1}^* =b_0^*$). $\tilde{b}_k$ is either smaller than $b_k^*$ (the middle panel) or greater than $b_k^*$ (the lower panel). In the middle panel, locales in the $\star$ region have a higher remoteness than their counterparts in the upper panel. In the lower panel, the size of state $k$ is too large for the overlord at locale $\tilde{b}_{k-1}$. Thus, in either case, state $k$ is not in equilibrium. \par}
\end{figure}

Since all overlords belong to $\{b_n^*\}$ and no other overlords exist, the world partition we have identified is the only equilibrium partition. These overlords accept offers from themselves and from some of their distal neighbors. We can formally prove:
\begin{proposition}\label{prop:existunique}
There exists a unique equilibrium world partition.
\end{proposition}
\begin{proof}
Existence of the equilibrium partition follows from the earlier solving process, starting from the WGC and moving outward. Uniqueness of the equilibrium partition is established in Appendix~\ref{uniqueproof}.
\end{proof}

Finally, given $\{b_n^*\}$, consumption choices $C^z(t)$ and $C^l(t)$ are determined, completing the unique equilibrium $\Omega$.

We have established the existence and uniqueness of the equilibrium. We refer to the model developed so far as the baseline model. In the next three sections, we apply the baseline model to three thematic areas: political economy, political security, and political ideology.

\section{Political Economy\label{sec:polecon}}

\subsection{Interstate trade}

The empirical trade literature finds that interstate trade typically follows a gravity pattern: two states exhibit larger bilateral trade volumes when they are larger in size and closer in distance. To rationalize this empirical regularity, we derive a gravity equation from our model:
\begin{proposition}\label{prop:gravity}
The exports from state $m$ to state $n$ follow 
\begin{equation}
X_{m,n} = \zeta {S_m}{S_n}\exp (  - \tau D_{m,n}), \label{eq:grav}
\end{equation}
where $\zeta$ is a positive scalar and $D_{m,n} \equiv \min_{\,s \in m, t \in n} d(t,s)$ is the shortest locale-to-locale distance between the two states.
\end{proposition}
\begin{proof}
See Appendix \ref{gravproof}. 
\end{proof}

\Cref{prop:gravity} connects our model to the gravity model literature.\footnote{See \citet{Anderson11}, \citet{HM14}, and \citet{AAT19} for reviews.} Early gravity models took the form of equation~\eqref{eq:grav}, which is analogous to Newton’s law of universal gravitation. These early formulations were incomplete because they did not account for the differential remoteness of trading partners from the rest of the world. Trade between two states depends not only on their bilateral distance but also on their distances from all other potential trade partners, since the importance of bilateral distance varies with each state’s overall geographic position. To capture these general equilibrium effects, gravity models developed over the past two decades incorporate explicit remoteness terms, commonly referred to as multilateral resistance. Our gravity equation embeds these general equilibrium forces while preserving the original Newtonian structure. 

Since borders are not fixed in our model, the gravity equation \eqref{eq:grav} can be used to analyze how parameters influence the interaction between interstate trade and state borders. In this sense, our gravity equation represents a long-run relationship, whereas the gravity equations in the existing literature are short-run. Our long-run perspective, which endogenizes both trade and borders, delivers the following implications. Below, we use $\hat v = dv/v$ to denote the percentage change in any variable $v$.

First, consider a decrease in \(\tau\) such that \(d\tau<0\). Its impact on bilateral trade volume \(X_{m,n}\) can be decomposed into three effects:
\begin{equation}
\underbrace{\hat{X}_{m,n}}_{\lesseqgtr 0}
=
\underbrace{\hat{S}_m + \hat{S}_n}_{\text{size effect}<0}
\underbrace{-D_{m,n}\, d\tau}_{\text{direct effect}>0}
\underbrace{-\tau\, dD_{m,n}}_{\text{location effect}>0}.
\label{eq:grav4}
\end{equation}
Among the three effects in \eqref{eq:grav4}, the \textit{direct effect} is self-explanatory. The \textit{size effect} captures the fact that both states shrink in size when \(\tau\) decreases and smaller states trade less. The net of the direct and size effects is therefore ambiguous, depending on their relative magnitudes. A third \textit{location effect} further contributes to this ambiguity: as a reduction in  \(\tau\) leads to smaller states worldwide, the contraction of the intermediate states between \(m\) and \(n\) shortens the effective distance between them and thus raises trade volume.

Next, consider an increase in \(h\) such that \(dh>0\). Its impact on bilateral trade volume \(X_{m,n}\) can be decomposed into two components:
\begin{equation}
\underbrace{\hat{X}_{m,n}}_{\lesseqgtr 0}
=
\underbrace{\hat{S}_m + \hat{S}_n}_{\text{size effects}<0}
\underbrace{-\tau\, dD_{m,n}}_{\text{location effect}>0}.
\label{eq:grav4b}
\end{equation}
The size effect is negative because state sizes shrink as \(h\) increases. The location effect is positive because locations \(m\) and \(n\) move closer to each other when the intermediate states contract. As a result, the net effect on trade is ambiguous. This is reminiscent of historical periods in which the collapse of empires did not necessarily reduce trade, as the removal of intermediate empires unblocked trade routes and counteracted the decline in trade.

A variant of equation \eqref{eq:grav} can be derived to characterize trade between two fixed areas. Fixed areas are geographic intervals anchored to given coordinates. Fixed areas never move, whereas the borders of a state with a given index may shift so that the state fully or partially contains a fixed area, or is fully contained within it. Trade between fixed areas \(u\) and \(v\) follows
\begin{equation}
X_{u,v} = \zeta F_u F_v \exp\!\left(-\tau D_{m(u),n(v)}\right),
\label{eq:grav_F}
\end{equation}
where \(F_u\) and \(F_v\) denote the sizes of the two fixed areas, and \(D_{m(u),n(v)}\) is the shortest distance between the states containing them. \Cref{fig:fixedareas} illustrates the two fixed areas and the distance over which trade costs apply—namely, the segment between the borders of their respective states that face each other. When \(\tau\) decreases or \(h\) increases, states shrink in size while \(F_u\) and \(F_v\) remain unchanged. The contraction of the intermediate state(s) shortens the cost-ridden distance between the two fixed areas, thereby increasing trade between them. This mechanism operates regardless of whether the two fixed areas lie in the same hemisphere or in different hemispheres. Technically, when fixed areas rather than states are the trading entities, the size effects in equations \eqref{eq:grav4}--\eqref{eq:grav4b} are absent, so the trade response has an unambiguous sign.

\begin{figure}[h!]
\centering
\caption{Trade between Two Fixed Areas when States Move}
\label{fig:fixedareas}
\hspace*{0cm}\includegraphics[page=2,scale=.7,clip]{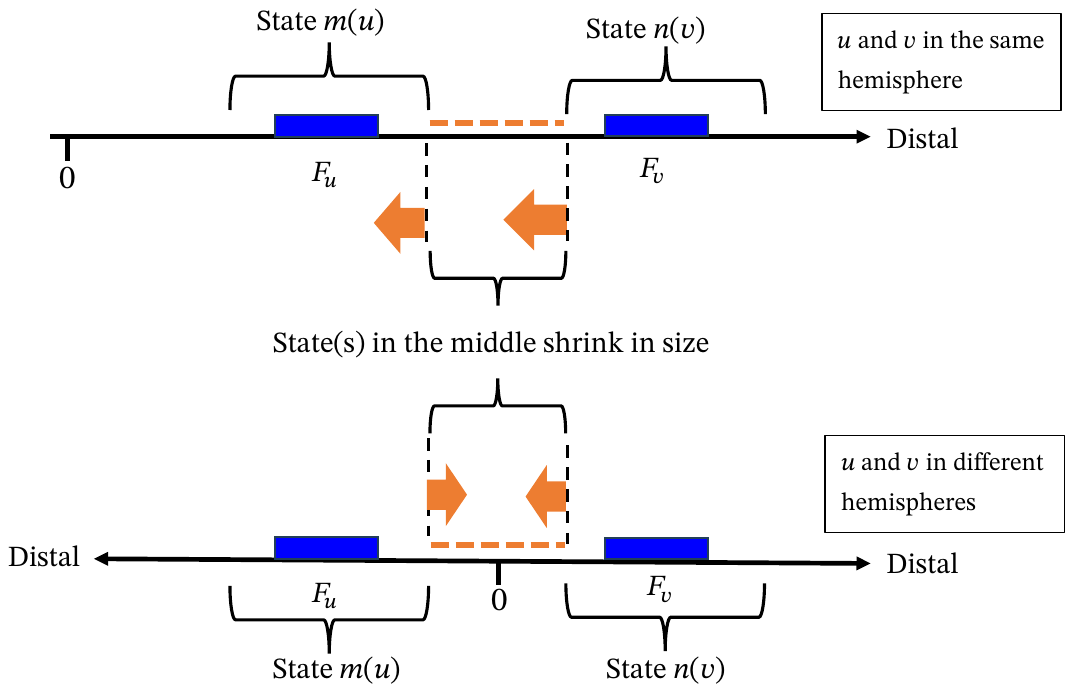}
\vspace{-300pt}
{ \par \raggedright \footnotesize \textit{Notes:} Trade costs between fixed areas \(F_u\) and \(F_v\) apply only beyond their states’ mutual border—namely, along the orange dashed segment. As \(\tau \downarrow\) or \(h \uparrow\), this segment shortens because the intermediate states contract in size. This result holds whether the two fixed locations lie in the same hemisphere (top panel) or in different hemispheres (bottom panel). \par}
\end{figure} 

Notice that if \(\tau\) continues to decrease (or \(h\) continues to increase), the state affiliation of the two fixed areas will change, causing \(X_{u,v}\) to fluctuate with the movement of the relevant borders. As this process continues, the states overlapping with the two fixed areas, as well as the intermediate states between them, eventually become infinitesimal. Throughout this process, \(F_u\) and \(F_v\) in equation \eqref{eq:grav_F} remain unchanged, because locale economies have invariant economic sizes under the Cobb--Douglas production and consumption structures. All fluctuations in trade arise from the location effect, which gradually vanishes as \(D_{m(u),n(v)} \to D_{u,v}\). In the limit, the physical distance between the two fixed areas is occupied by a continuum of infinitesimal states, so that trade costs apply along the entire distance. Consequently, the fixed-area gravity equation \eqref{eq:grav_F} converges to a state-free gravity equation: $X_{u,v} = \zeta F_u F_v \exp(-\tau D_{u,v})$.

\subsection{Interstate migration}
Just as goods are traded across states, production factors may also migrate across borders. In our model, labor has incentives to migrate between states in response to real wage differentials. Borders introduce frictions to labor mobility analogous to those affecting interstate trade in goods. In this subsection, we model interstate migration by allowing labor to move freely between adjacent states. 

Consider two adjacent states $m$ and $n$ in the right hemisphere, where state $m$ lies farther from the WGC than state $n$ (so $ m=n+1$). Suppose that the border between them no longer restricts labor mobility. Aggregating equation~\eqref{eq:factorincomeL} to the state level yields
\begin{equation}
w_n L_n \;\equiv\; \int_{b_{n-1}}^{b_n} w(t) \, l(t)\, dt = \int_{b_{n-1}}^{b_n} (1-\alpha)\kappa \, dt = (1-\alpha)\kappa S_n,\label{eq:statelaborpay}
\end{equation}
for state $n$, where $L_n$ denotes total labor employment in state $n$. An analogous expression holds for state $m$. Let $w_n$ and $w_m$ denote post-migration nominal wages. Free migration between the two states implies real wage equalization:
\begin{equation}
\frac{w_n}{R_n} = \frac{w_m}{R_m}.\label{eq:realwageequal}
\end{equation}
Because the model implies a lower initial real wage in state $m$ than in state $n$, labor migrates from state $m$ to state $n$. Combining equations~\eqref{eq:statelaborpay} and~\eqref{eq:realwageequal}, we obtain
\begin{equation}
\frac{S_n}{S_m} \cdot \frac{L_m - M_{m,n}}{L_n + M_{m,n}}=\frac{R_n}{R_m},\label{eq:migrationflowprepare}
\end{equation}
where $M_{m,n}$ denotes the migration flow from state $m$ to state $n$.

Normalize the initial labor supply of each locale $s$ (i.e., $l^0(s)$) to one, then we have:
\begin{proposition}\label{prop:migration}
The equilibrium migration flow from state $m$ to its proximal adjacent state $n$ is given by
\begin{equation}
M_{m,n}=S_m S_n \frac{\Phi_n-\Phi_m}{\Phi_n+\Phi_m},\label{eq:migrationflow}
\end{equation}
where
\[
\Phi_k=(1-b_k)^{\frac{1}{\gamma-1}},
\qquad k\in\{m,n\}.
\]
\end{proposition}

\begin{proof}
See \Cref{migrationproof}.
\end{proof}

\Cref{prop:migration} delivers three implications. First, the terms $\Phi_n$ and $\Phi_m$ capture the real wage differential between two adjacent states. Importantly, holding $S_n$ and $S_m$ fixed, a large real wage differential alone does not guarantee a large migration flow. The differential is normalized by the overall locational advantages of the two states, given by $\Phi_n+\Phi_m$. As a result, for two states that are remote from the rest of the world, a smaller real-wage differential can induce a migration flow of the same magnitude than between two states located close to the WGC.

Second, equation~\eqref{eq:migrationflow} fully accounts for the production and consumption adjustments induced by migration from state $m$ to state $n$. Because production and consumption are equalized across locales within each state, migrants from state $m$ are evenly distributed across locales in state $n$. The resulting changes in labor supply affect nominal wages in the two states:
\[
w_n^{new} < w_n^{old}, \qquad w_m^{new} > w_m^{old},
\]
and the prices of the locale-level output of the two states adjust accordingly:
\[
p(s\in n)^{new} < p(s \in n)^{old}, \qquad p(s\in m)^{new} > p(s \in m)^{old}.
\]
However, under Cobb--Douglas production and consumption structures, locale-level output and consumption remain unchanged in both states. In fact, apart from the induced input and output price adjustments, all other aspects of the equilibrium remain unchanged.

Third, equation~\eqref{eq:migrationflow} continues to hold even when states  \(m\) and \(n\)  are not adjacent. In this case, the derivation above still applies, and the migration flow remains \(M_{m,n}\). That said, distance-related migration frictions may become relevant, since moving to a more distant state is plausibly more costly than moving to a nearby one. Incorporating a distance-related term would render equation~\eqref{eq:migrationflow} fully gravity-like. However, we believe that interstate migration is much more limited in scope than interstate trade; accordingly, we do not pursue this extension and retain equation~\eqref{eq:migrationflow} in its current half-gravity form.

\subsection{Border effect}
A given border matters relatively more for the state on its geographically disadvantaged side. This asymmetry, observed as an empirical regularity and known as the border effect, has been widely documented in the international trade and macroeconomics literatures \citep{McCallum95, ER96, Helliwell98, AVW04}. Our model rationalizes the border effect in a parsimonious way, without explicitly modeling trade volumes or prices:
\begin{proposition}\label{prop:sensitivity}
Between two adjacent states, the shared border matters more for the distal state than for the proximal state, and this asymmetry increases with the states’ distance from the WGC.
\end{proposition}

\begin{proof}
Consider three borders in the right hemisphere: \(b_{k-1}\), \(b_{k}\), and \(b_{k+1}\), with
\(0 \le b_{k-1} < b_k < b_{k+1} \le 1\).
Locales in \([b_{k-1}, b_k)\) constitute state \(k\), and locales in \([b_k, b_{k+1})\) constitute state \(k+1\).

Define relative importance in magnitude as
\begin{equation}
\widetilde T_k
\equiv
\frac{\left|\frac{\partial R_{k+1}}{\partial b_k}\right| / R_{k+1}}
     {\left|\frac{\partial R_k}{\partial b_k}\right| / R_k}
=
\frac{1+b_k}{1-b_k},
\label{eq:whichside}
\end{equation}
where the equality follows from equation~\eqref{eq:remoteness}. It is immediate that \(\widetilde T_k>1\). Moreover,
\[
\frac{\partial \widetilde T_k}{\partial b_k}
=
\frac{2}{(1-b_k)^2}
>
0.
\]
\end{proof}

Our model not only illustrates the border effect both geometrically and analytically, but also extends it. In particular, we find that the border effect intensifies when the two states are jointly geographically disadvantaged, as \( \frac{\partial \widetilde T_k}{\partial b_k} > 0 \). That is, when states $k$ and $k+1$ lie farther from the WGC, the same border has a greater impact on the two adjacent states, resulting in a larger border effect.

\section{Political Security\label{sec:polgeo}}

\subsection{Border stability}
Forces of globalization can be divided into an economic type, $d\tau<0$ (lower foreign trade costs), and a political type,  $dh<0$ (greater acceptance of others as fellows). We now examine how such globalization shocks reshape the political geography of the linear world by altering its partition:
\begin{proposition}\label{prop:borderstability}
When either $\tau$ or $h$ changes, border invariance requires the other to change in the same direction. Otherwise, interior borders are subject to greater pressure to adjust.
\end{proposition}
\begin{proof}
For expositional purposes, we treat $h$ implicitly as a function of $\tau$ and obtain 
\[
\frac{dh}{d\tau} > 0, \quad \frac{\partial (dh/d\tau)}{\partial b_n}<0.
\]
See Appendix \ref{dhdt} for details.
\end{proof}
\noindent Intuitively, when trade costs decline, locales must tolerate larger state sizes in order to preserve the current world partition. Absent this increased tolerance, the existing partition collapses and smaller states emerge. In other words, economic globalization and political globalization must co-occur to avoid fragmentation of the world.

When economic and political globalization are unbalanced, collapse of the world partition initiates at the WGC, where territories are most valuable. At first glance, this appears to contradict \Cref{prop:sensitivity}, which states that borders farther from the WGC are more “sensitive.” The two propositions, however, concern distinct phenomena: \Cref{prop:borderstability} compares border stability across the world, whereas \Cref{prop:sensitivity} compares the consequentiality of border changes. There is no contradiction—distal borders are less likely to change in response to parameter changes, but have larger effects when they do. For this reason, discussions of border sensitivity must distinguish between sensitivity in terms of stability and sensitivity in terms of consequentiality. These two notions exhibit opposite spatial patterns.

\subsection{State sizes}
In our model, states located farther from the WGC are more remote from the rest of the world. To offset this remoteness, they must rely on a larger domestic market in order to keep remoteness $R_n$ low:
\begin{proposition}\label{prop:size}
Provided that $\tau$ is not too low, the size of state $n \neq 0$ increases with $|n|$.
\end{proposition}
\begin{proof}
See Appendix \ref{hypo}.
\end{proof}

One might expect state 0 to be the smallest state in the world, since it is closest to the rest of the world. However, state 0 is not necessarily the smallest state because a competing force applies uniquely to it. Unlike all other states, state 0 expands its borders in two opposite directions. As a result, the marginal return relative to the disutility from a unit increase in size is higher for state 0 than for other states. To illustrate this difference, compare state $0$'s first-order condition \eqref{eq:FOC00} with state $1$'s first-order condition \eqref{eq:FOC11}:
\begin{equation*}
\tau R_0^{\gamma - 1}\left(1 - S_0\right) = h
\hspace{5mm} \text{versus} \hspace{5mm}
\tau R_1^{\gamma - 1}\left(1 - \frac{S_0}{2} - S_1\right) = h.
\end{equation*}
Recall $R_0 < R_1$ and $\gamma > 1$. The only requirement on the relative sizes of $S_0$ and $S_1$ is that
$1 - S_0 > 1 - S_0/2 - S_1$, or $S_0/2<S_1$. $S_0$ is not necessarily smaller than $S_1$. 

The exceptional role of state 0 in terms of size is reminiscent of historical empires, which were often located near the center of continents or subcontinents. The size of state 0 varies with $\tau$, $h$, $\gamma$, and their combination. The resulting distribution of state sizes over time can be analyzed using comparative statics across repeated, unique equilibria of the model. In this perspective, fluctuations in the size of state 0 can themselves be interpreted as shocks. This repeated-equilibrium view leads to the following result:
\begin{proposition}\label{prop:state0}
Provided that $\tau$ is not too low, an increase in the size of state $0$ makes state $n \neq 0$ farther from the WGC and larger in size, and this effect increases with $|n|$.
\end{proposition}
\begin{proof}
See Appendix \ref{state0proof}.
\end{proof}

\begin{figure}[h!]
\centering
\caption{Illustration of \Cref{prop:state0} Using State~1 as an Example}
\label{fig:statesize}
\hspace*{1cm}\includegraphics[page=3,scale=.8,clip]{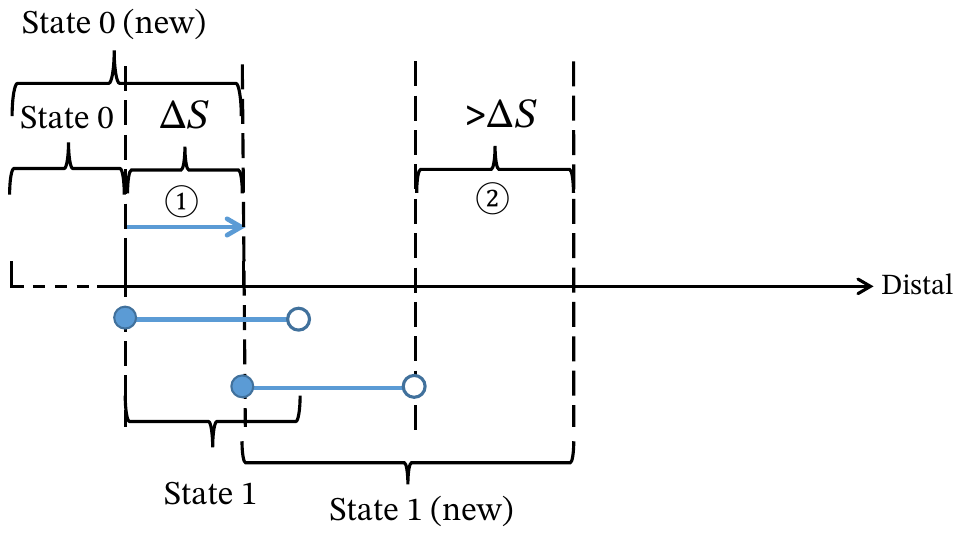}
\vspace{-400pt}
{ \par \raggedright \footnotesize \textit{Notes:} Solid (hollow) circles denote closed (open) interval endpoints. Locales represented by solid (hollow) circles belong to the distal (proximal) state. \par}
\end{figure}

The intuition behind \Cref{prop:state0} is as follows. State 0 is always the state with the lowest $R$ among all contemporaneous states, so surrounding locales prefer to join state 0 whenever feasible, and this preference propagates outward. When state 0 increases in size, other states are effectively ``pushed away'' from the WGC. As they are pushed away, these states must grow in size because the new locales they acquire are less advantageous than those they relinquish. \Cref{fig:statesize} illustrates this mechanism using state 1 as an example. Suppose the size of state 0 increases toward the right by a distance $\Delta S$. Locales in region \circled{1}, with measure $\Delta S$, which previously belonged to state 1, are now absorbed by state 0. Consequently, state 1 must expand to include region \circled{2}, whose measure exceeds $\Delta S$. This increase in size occurs because the newly acquired territory \circled{2} consists of less favorable locations than the lost territory \circled{1}. The same force applies to state 2, but with greater magnitude.

\subsection{Suffrage expansion}
Labor does not bear governance costs; consequently, incorporating labor into border drawing leads to a different world partition.\footnote{Economists across the political spectrum have long noted labor’s limited interests in national politics. For example, \citet{Marx48} argues that workers “have no country” because they lack a material stake in any state, whereas \citet{Friedman62} notes that workers move wherever their labor has the most productive use, rendering their nationality largely obsolete.} This mechanism resembles the effect of suffrage expansion: extending political voice to labor alters how the state is governed. Formally, consider a border setter at locale \(t\) with objective function:
\begin{equation}
W(t;\phi)=U(t)+\phi V(t), \label{eq:se}
\end{equation}
where \(\phi>0\) denotes the weight that the border drawer assigns to labor.\footnote{\(\phi\) is distinct from the \(\psi\) in \(V(t)\): the latter governs the marginal utility of consumption and adds flexibility to the model’s specification.
}  
Let border drawers determine state sizes, denoted by \(S_n^{SE}\), where \(SE\) refers to suffrage expansion. Let \(N^{SE}\) denote the corresponding number of states, and let \(N^*\) denote the number of states in the baseline model. We then have:
\begin{proposition} \label{prop:SE}
\(N^{SE} \le N^*\), and \(S_n^{SE} > S_n^*\) for each state \(n\) if \(N^{SE}=N^*\).
\end{proposition}
\begin{proof}
See Appendix \ref{se}. 
\end{proof}
\noindent In other words, granting voice to labor implies either (a) a smaller total number of states globally or (b) larger individual states for any given index $n$, holding the number of states fixed. The intuition is straightforward: when borders are determined solely by lords, they fail to internalize the positive consumption externality that larger states confer on labor.

Three notes are in order. First, when \(\phi=1\), the suffrage-extension solution coincides with the solution chosen by the social planner of each locale, who weights the welfare of the lord and the labor equally. By comparison, the solution to the baseline model is a decentralized equilibrium in which borders are drawn by lords, who can be viewed as border drawers with $\phi=0$.

Second, the results in \Cref{prop:SE} become stronger as \(\tau\) decreases, especially for states located farther from the WGC:
\[
\frac{\partial}{\partial\tau}\bigl(S_n^{\text{SE}}-S_n^*\bigr)<0,\qquad 
\frac{\partial^2 \bigl(S_n^{\text{SE}}-S_n^*\bigr)}{\partial \tau\,\partial R_n}<0.
\]
Intuitively, when \(\tau\) is high, borders determined under suffrage extension and those determined solely by lords are similar, because labor’s consumption gains from belonging to a larger state are limited—particularly in peripheral states. Therefore, a reduction in \(\tau\), interpreted as economic globalization, intensifies the conflict of interest between lords and labor in geographically disadvantaged locations. 

Third, state formation reflects both cross-locale compromises (a ``republican problem,'' driven by \(h\)) and within-locale compromises (a ``democratic problem,'' driven by \(\phi\)).\footnote{James Madison, in \textit{Federalist No.~10}, contrasts a democracy—``a society consisting of a small number of citizens, who assemble and administer the government in person''—with a republic—``a government in which the scheme of representation takes place.'' This distinction parallels the within-locale democratic and cross-locale republican structure modeled by us.} Our model, as extended in this subsection, integrates two distinct forces in modern politics that underpin democracy without reference to any specific political regime: the state rests on political tolerance among constituents and on a balance of economic interests within each constituent.

\section{Political Ideology\label{sec:polideo}}

\subsection{National opinion}
The baseline model provides a clean framework for studying interstate ideology. Because locales share identical factor endowments, technology, and preferences, ideological differences across locales must arise from geographic location. This structure enables us to build a variety of models of political ideology with geographic determinism.

Suppose the world line also represents a spectrum of opinions on a given issue, with the most left-wing (right-wing) opinion held by locale $t=-1$ ($t=1$). We define a state’s national opinion as the opinion held by its median locale. This definition allows us to compare national opinions across states and relative to the global median.

For a state $n$ located in the right hemisphere of the world, its national opinion is
\begin{equation}
G_{n}
= \frac{1}{2} \left( b_{n-1} + b_{n} \right)
= b_{n-1} + \frac{S_{n}}{2}
= \frac{1}{2} S_0 + \sum_{k=1}^{n-1} S_k + \frac{1}{2} S_{n}.
\end{equation}
The national opinion of state $-n$ in the left hemisphere follows immediately: $G_{-n} = -G_n$. Likewise, the national opinion of the central state is $G_0 = 0$. We can then use the moments of $G_n$ to characterize the distribution of national opinions. The mean is
\begin{equation}
\mathbb{E}(G_n)
= \frac{1}{2N} \sum_{n=-N}^{N} G_n
= \frac{1}{2N} \sum_{n=1}^{N} \left( G_n + G_{-n} \right)
= 0,
\end{equation}
and the variance is
\begin{equation}
\text{Var}(G_n)
= \frac{1}{2N} \sum_{n=-N}^{N} \left( G_n - \mathbb{E}(G_n) \right)^2
= \frac{1}{N} \sum_{n=1}^{N} G_n^2.
\label{eq:varG}
\end{equation}
\noindent Intuitively, the median of national medians coincides with the global median. We take the global median as the world’s opinion, so that a state whose national opinion lies to the left (right) of the world’s opinion is classified as left-leaning (right-leaning). The dispersion of national opinion, as measured by $\text{Var}(G_n)$, captures the degree of differentiation across national opinions. As equation~\eqref{eq:varG} shows, larger deviations of national opinions from the global median reduce global consensus at a quadratic rate.

We next study how global consensus responds to shocks in global economic and political integration:
\begin{proposition}\label{prop:consensus}
With the number of states fixed, national opinions diverge as economic globalization retreats ($\tau \uparrow$) but converge as political globalization retreats ($h \uparrow$).
\end{proposition}
\begin{proof}
Consider the right hemisphere without loss of generality:
\begin{align}
\frac{\partial \text{Var}(G_n)}{\partial \tau}
&= \frac{1}{N} \sum_{n=1}^{N} G_n
\left(
\frac{\partial S_0}{\partial \tau}
+ \overbrace{2\sum_{k=1}^{n-1}\frac{\partial S_k}{\partial \tau}}^{(\#)}
+ \frac{\partial S_{n}}{\partial \tau}
\right) > 0,
\label{eq:consensus_tau_shock} \\
\frac{\partial \text{Var}(G_n)}{\partial h}
&= \frac{1}{N} \sum_{n=1}^{N} G_n
\left(
\frac{\partial S_0}{\partial h}
+ \underbrace{2\sum_{k=1}^{n-1}\frac{\partial S_k}{\partial h}}_{(\#)}
+ \frac{\partial S_{n}}{\partial h}
\right) < 0.
\label{eq:consensus_h_shock}
\end{align}
See \Cref{consensusproof} for details.
\end{proof}
\noindent Intuitively, an increase in $\tau$ raises optimal state sizes in order to internalize trade costs, pushing states, and hence their national opinions, further apart. By contrast, an increase in $h$ raises governance costs, leading to smaller states and a tighter clustering of national opinions. The latter case is particularly interesting: although constituents within each state become less tolerant of one another, global consensus increases because the world political system is composed of states that are more similar to each other.

It is important to note that \Cref{prop:consensus} holds the number of states in the world fixed. Once the number of states adjusts endogenously, \Cref{prop:consensus} need not hold. On the one hand, the creation (destruction) of states increases (decreases) state differentiation by introducing (eliminating) additional national opinions. On the other hand, shocks that give rise to new states also reduce state sizes across the board, which compresses the spacing among national opinions. These opposing forces imply an ambiguous net effect on global ideological dispersion. This can be seen directly from the summation term $(\#)$ in equations~\eqref{eq:consensus_tau_shock}--\eqref{eq:consensus_h_shock}: when the number of states increases (decreases), the summation includes more (fewer) terms, but each term is smaller (larger) in magnitude.

\subsection{Regional separatism}

In the baseline model, a locale’s choice of which neighboring locales to join determines its remoteness from the rest of the world. However, a locale can join only those neighboring locales that are available to it. As a result, a locale may prefer a state affiliation different from its equilibrium one. We refer to such a counterfactual, idealized state affiliation as \emph{separatism}, which is the focus of the analysis below. Our approach to studying separatism mirrors how social scientists conceptualize societal sentiments: as imagined communities with particular compositions of fellows.\footnote{For example, economists’ modeling of religionism \citep{Iannaccone92}, political scientists’ modeling of nationalism \citep{Anderson83}, and sociologists’ modeling of racism \citep{Blumer58}.} In our setting, the imagined community corresponds to each locale’s ideal state. 

\begin{figure}[h!]
\centering
\caption{Separatism of Locale $t$}
\label{fig:separatism}
\hspace*{.5cm}\includegraphics[page=4,scale=.7,clip]{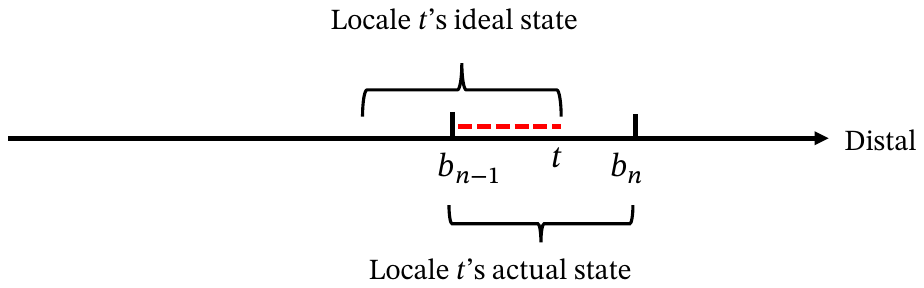}
\vspace{-400pt}
{ \par \raggedright \footnotesize \textit{Notes:} The red dashed line represents the overlap between locale $t$’s ideal state and its actual state. The overlap reduces locale $t$’s separatism. \par}
\end{figure}

Consider a locale in the right hemisphere. By equation \eqref{eq:FOC}, the ideal set of fellow locales that locale $t$ wishes to have is 
\[
I_t = [b_{n-1}^{I}(t), b_n^{I}(t)] 
= [t + \varepsilon - S_n,\; t + \varepsilon],
\]
which has the same size $S_n$ as its actual state $[b_{n-1}, b_n)$.\footnote{Here we assume $t + \varepsilon - S_n > 0$; that is, locales in the right hemisphere do not seek to join left-hemisphere locales when forming separatism.} We ignore the $\epsilon$ for expositional convenience. Let the state to which locale $t$ belongs be denoted by $n_t$ as before. The overlap between locale $t$'s ideal fellow locales and its actual fellow locales is 
\[
\Lambda(t)
= \frac{\bigl|I_t\cap n_t\bigr|}{S_n}
= 1 - \frac{|b_n - t| }{S_n},
\]
where the denominator normalizes the overlap so that $\Lambda(t)$ is comparable across $t$. $\Lambda(t)$ measures the ideality of locale's factual state $n_t$. \Cref{fig:separatism} illustrates the relationship among these notions. We define the level of separatism of locale $t$ as $\sigma(t) \equiv 1 - \Lambda(t)= \frac{|b_n - t|}{S_n}$, which lies between 0 and 1. Then we have
\begin{proposition} \label{prop:separatism}
Separatism within a state is stronger in proximal locales, forming a centripetal force that increases with $\tau$ and $R_n$.
\end{proposition}
\begin{proof}
Consider the right hemisphere such that $\sigma(t) = \frac{b_n - t}{S_n}$. It is easy to derive $\frac{d\sigma(t)}{dt} = -\frac{1}{S_n} < 0$. Then, for any locale $t\in[b_{n-1},b_n]$, define $\Delta \equiv t - b_{n-1}\in[0,S_n]$. Since $\frac{\partial S_n}{\partial \tau}
= \frac{h}{\tau^{2} R_n^{\gamma-1}} > 0$ and $\frac{\partial S_n}{\partial R_n}= \frac{h(\gamma-1)}{\tau R_n^\gamma} $, it can be verified that
 \[
\frac{\partial \sigma(t)}{\partial\tau}
= \Delta\,\frac{1}{S_n^2}\frac{\partial S_n}{\partial\tau}
= \frac{\Delta}{S_n^2}\cdot \frac{h}{\tau^2 R_n^{\gamma-1}}
>0,  \qquad \frac{\partial \sigma(t)}{\partial R_n}
= \frac{\Delta}{S_n^2}\cdot \frac{h(\gamma-1)}{\tau R_n^\gamma}
> 0.\]
\end{proof}

\Cref{prop:separatism} implies that locales with greater economic interests in the rest of the world exhibit stronger separatist tendencies.\footnote{The leftmost locale in a state in the right hemisphere exhibits the strongest separatism. This does not contradict its earlier role in governing state $n$ (i.e., serving as an overlord). Historically, ruling elites have often courted foreign powers, as in nineteenth-century Russia, Egypt, and Japan. In both the historical episodes and our model, political influence derives from proximity to the foreign polities with which they seek association.} Such economic interests may translate into invasive actions toward the neighboring state on the locale’s proximal side. Early geopolitical analyses, most notably the heartland theory \citep{Mackinder04}, emphasize the strategic importance of Eastern Europe for global security in the early twentieth century.\footnote{These analyses continue to influence modern works in international relations \citep[e.g.,][]{Morgenthau48, Kissinger94, Brzezinski97}.} In these frameworks, Eurasia plays a role analogous to the world line in our theory: because of its central location, Eastern Europe is of crucial strategic importance. Our model not only provides a microfoundation for this line of thought but also extends it. In particular, invasive separatism is stronger in states that are more remote from the rest of the world and weakens as trade frictions $\tau$ fall. Moreover, as $\tau$ falls, invasive separatism in remote states decreases more sharply than elsewhere, a result that follows immediately from the cross-partial derivative $\frac{\partial^2 \sigma(t)}{\partial R_n\,\partial \tau}<0$.

\section{Concluding Remarks \label{sec:conclude}}

Geopolitics centers on two themes: trade and borders. Modeling them jointly is challenging because trade across states is interdependent, borders across states are interdependent, and trade and borders are themselves mutually interdependent. We address this challenge by adopting linearization, a traditional approach in economics. Building on a linear world, the model developed here integrates trade and borders within a single general equilibrium. Through their interaction, trade and borders generate rich geopolitical behaviors both within and across states, providing a unified and tractable framework for analyzing political economy, security, and ideology within and across states.

The elegance and versatility of the model rest on two key simplifications. The first is the use of Cobb--Douglas production and consumption structures. They transform an otherwise complex coordination game---locales organize themselves into states---into a tractable framework in which locales unilaterally select fellow locales. Under CES technology, the model would remain unchanged only if technologies were equally efficient across locales. Under CES consumption, the locational advantages of proximal locales that drive the model would be amplified: in addition to lower consumption costs, proximal locales would also earn higher income due to their larger share of world expenditure. In that case, the constant numerator in $1/R(t)$ would be replaced by an intricate term inversely related to $R(t)$, substantially reducing the model’s tractability.

The second simplification is the linear world assumption. In our model, linear geography guarantees the existence of a unique equilibrium under general conditions. Moreover, borders can be characterized analytically as points on the world line, allowing their behavior to inform a wide range of geopolitical phenomena. In \Cref{sec:network}, we instead model locales as nodes in a network and states as connected components, allowing for arbitrary geography. We show therein that additional structures are required to recover a unique world partition. More importantly, because borders lack analytical characterizations, border behavior cannot be used to analyze geopolitical phenomena in that setting.



\bibliographystyle{te}
\bibliography{benpenglong_bib}

\appendix

\newpage

\begin{center}
\section*{``A Linear Model of Geopolitics'' \\ Appendices \\ {\normalsize Ben G. Li and Penglong Zhang} \\ {\normalsize \today}}
\end{center}

\renewcommand{\thesection}{A}
\numberwithin{equation}{section}
\setcounter{equation}{0}  
\renewcommand{\thetable}{A}
\renewcommand{\thefigure}{A}
\setcounter{figure}{0} \renewcommand{\thefigure}{A\arabic{figure}}
\setcounter{table}{0} \renewcommand{\thetable}{A\arabic{table}}

\subsection{Proofs and derivations}

\subsubsection{Derivation of equation \eqref{eq:aggC}}\label{aggCproof}
On Date 3, at locale $t$, the lord’s first-order condition is
\begin{equation}
p(t,s)c^z(t,s)= \frac{C^z(t)^{1-\gamma}}{\lambda^z(t)} \equiv \kappa^z(t).
\label{eq:lordfoc}
\end{equation}
Substituting this condition into the lord’s budget constraint yields $\kappa^z(t) = r(t)z(t)/2$. Likewise, labor’s first-order condition is
\begin{equation}
p(t,s)c^l(t,s)
= \frac{\psi\, C^l(t)^{1-\gamma}}{\lambda^l(t)}
\equiv \kappa^l(t),
\label{eq:laborfoc}
\end{equation}
which implies $\kappa^l(t) = w(t)l(t)/2$. Thus, the aggregate first-order condition is obtained by summing equations
\eqref{eq:lordfoc} and \eqref{eq:laborfoc}:
\[
p(t,s)\, c(t,s)
= \frac{C^z(t)^{1-\gamma}}{\lambda^z(t)}
+ \frac{\psi C^l(t)^{1-\gamma}}{\lambda^l(t)}
\equiv \kappa(t).
\]

\subsubsection{Proof of \Cref{lemma:suff}\label{suffproof}}
By equation \eqref{eq:lordfoc}, we have $c^z(t,s)=\kappa^z(t)/p(t,s)$. By inserting the $c^z(t,s)$ into $C^z(t)$, we obtain
\begin{eqnarray*}
C^z(t) &=& \exp \left(\int_{-1}^1 (\ln \kappa^z(t)-\ln p(t,s)) ds \right) \\
       &=& \exp \left(\int_{-1}^1 (\ln \kappa^z(t)/p-\ln d(t,s)) ds\right) \\
       &=& \exp \left(2\ln \kappa^z(t)/p-\int_{-1}^1 \ln d(t,s) ds\right) \\
       &=& (\kappa^z(t)/p)^2\exp \left(-\int_{-1}^1 \ln d(t,s) ds\right) \\
       &=& \left(\frac{rz}{2p}\right)^2\exp \left(-\int_{-1}^1 \ln d(t,s) ds\right)\\
       &=& \left(\frac{rz}{2p}\right)^2/R(t),
\end{eqnarray*}
where $p=rz/2$ is the normalized factory-gate price. Thus,  $C^z(t) =1/R(t)$.

\subsubsection{Proof of \Cref{prop:existunique}: Uniqueness}\label{uniqueproof}

Assume that there exists another equilibrium world partition
\begin{equation}
\{\tilde b_{n}\} \equiv 
\{\tilde b_{-\tilde N},...,\tilde b_{-1},\tilde b_{-0},\tilde b_{0},\tilde b_{1},...,\tilde b_{\tilde N}\},\label{eq:newpartitionequ}
\end{equation}
that differs from $\{b_{n}^*\}$. Take the right hemisphere. Consider lord (at) $t=0$ who draws borders $\tilde b_{-0}$ and $\tilde b_0$. With the previously known equilibrium partition $\{b_{n}^*\}$,  
\[
\frac{\partial U(t=0|-b_0,b_0)}{\partial b_0} |_{b_0=b_0^*} = 0, 
\]
which is the first order condition for lord $t=0$. Thus lord $t=0$ who draws borders $\tilde b_{-0}$ and $\tilde b_0$ would be worse off in this scenario than in $\{b_{n}^*\}$ and deviate from this scenario to $\{b_{n}^*\}$.

Then, given $\tilde b_0=b_0^*$, suppose that there exists a such $k \ge 1$ that $\tilde{b}_{n}=b_{n}^*$  for $n=0,1,...,k-1$ and \(\tilde{b}_{k} \neq b_{k}^*\).

\begin{itemize}
\item If $\tilde b_k < b_k^*$, consider lord $t=\tilde b_k $. His remoteness would be based on his own locale, which has a higher remoteness than locale $\tilde b_{k-1}=b_{k-1}^*$ according to equation \eqref{eq:Rleftborder}. So, he would be worse off in this scenario than in $\{b_{n}^*\}$ and deviate from this scenario to $\{b_{n}^*\}$. If he insists on accepting offers made by lords $(\tilde b_{k},b_k^*)$, his acceptance would be invalid because those locales also made offers to lord $\tilde b_{k-1}$, who accepts their offers. Recall that when multiple offers made by a lord are accepted by his proximal lords, only the most proximal one takes effect. 

\item If $\tilde b_k > b_k^*$, consider lord $t=\tilde b_{k-1}$. According to equation \eqref{eq:FOC},
\[
U(t|\tilde b_{k-1},\tilde b_{k})=U(t| b_{k-1}^*,\tilde b_{k})<U(t|b_{k-1}^*,b_k^*).
\] 
Thus, he would be worse off in this scenario than in $\{b_{n}^*\}$ and deviate from this scenario to $\{b_{n}^*\}$.
\end{itemize}

In either case, a contradiction rises and thus there does not exist a such state $k$.

\subsubsection{Proof of \Cref{prop:gravity}}\label{gravproof}
Assume without loss of generality that $n \ge m+1 >0$, so that the two states lie in the right hemisphere.
Moreover, assume that state $n$ is farther from the WGC than state $m$, so that $D_{m,n} = b_{n-1} - b_m$. The export volume from state $m$ to state $n$ is 
\begin{eqnarray*}
{X_{m,n}} &=&  {S_m}\int_{{b_{n-1}}}^{{b_{n}}} {p(s)c(s,{b_m})d} s = \frac{S_m}{2}\int_{{b_{n-1}}}^{{b_{n }}} {\kappa d{{(s,{b_m})}^{ - 1}}d} s\\
&=& \frac{\kappa}{2\tau}{S_m}[\exp \left(  - \tau ({b_{n-1}} - {b_m})\right)  - \exp \left(  - \tau ({b_{n }} - {b_m})\right) ]\\
&=& \frac{\kappa}{2\tau}{S_m} \exp \left(-\tau D_{m,n}\right) \times (1-\exp \left(-\tau S_n\right)).
\end{eqnarray*}
Here, the second equality stems from equation \eqref{eq:localegdp}. Since state sizes are small relative to one (the size of each hemisphere), we have $1-\exp \left(-\tau S_n\right)=\tau S_n$. Thus,
\[
{X_{m,n}}=\zeta S_m S_n\exp \left(-\tau D_{m,n}\right),
\]
where $\zeta \equiv \kappa/2$ applies to all pairs worldwide.

\subsubsection{Proof of \Cref{prop:migration}}\label{migrationproof}
Equation \eqref{eq:FOC} implies
\begin{equation}
\frac{R_n}{R_{m}} = \left(\frac{1-b_{m}}{1-b_{n}}\right)^{\frac{1}{\gamma-1}}.\label{eq:Rratio}
\end{equation}
Inserting equation \eqref{eq:Rratio} into equation \eqref{eq:migrationflowprepare} gives
\[
\frac{S_n}{S_{m}} \frac{S_{m}-M_{m,n}}{S_n+M_{m,n}} = \left(\frac{1-b_{m}}{1-b_{n}}\right)^{\frac{1}{\gamma-1}}.
\]
Here, note that $L_n=S_n$ because the initial labor supply of each locale $s$ (i.e., $l^0(s)$) has been normalized to one. Rearrange the terms to solve for $M_{m,n}$: 
\[
M_{m,n} = S_n S_m \frac{\Phi_n - \Phi_m}{\Phi_n + \Phi_m}, \text{ where } \Phi_{k=m\text{ or }n} = (1-b_{k})^{\frac{1}{\gamma-1}}.
\]

\subsubsection{Proof of \Cref{prop:borderstability}}\label{dhdt}

The first-order condition \eqref{eq:FOC} for state $n$   is equivalent to 
\begin{equation}
F \equiv \tau R_n^{\gamma  - 1}(1 - {b_{n - 1}} - {S_n}) - h=0,
\end{equation}
which implies the following partial derivatives:
\begin{eqnarray}
{F_h} &=&  - 1 < 0,\label{eq:Fh} \\ 
{F_S} &=&  - (\gamma  - 1)R_n^{\gamma  - 1}{\tau ^2}{(1 - {b_n})^2} - \tau R_n^{\gamma  - 1} < 0,\label{eq:FS}\\
{F_{b_{n-1}}} &=& (\gamma  - 1)R_n^{\gamma  - 1}{\tau ^2}(2{b_{n - 1}} + {S_n})(1 - {b_{n - 1}} - {S_n}) - \tau R_n^{\gamma  - 1}, \label{eq:Fb}\\ 
{F_{\tau}} &=& R_n^{\gamma  - 1}(1 - {b_n})\left( 1 + \frac{\tau (\gamma  - 1)}{2}[{(1 + {b_{n - 1}})^2} + {(1 - {b_n})^2}]\right)  > 0. \label{eq:Ftau}
\end{eqnarray}
It follows that
\begin{equation*}
\frac{dh}{d\tau}=-\frac{F_{\tau}}{F_h}=F_{\tau} = R_n^{\gamma  - 1}(1 - {b_n})\left( 1 + \frac{\tau (\gamma  - 1)}{2}[{(1 + {b_{n - 1}})^2} + {(1 - {b_n})^2}]\right)  > 0,
\end{equation*}
and 
\begin{equation*}
\frac{\partial (dh/d\tau)}{\partial b_n}<0.
\end{equation*}

\subsubsection{Proof of \Cref{prop:size}}\label{hypo}
Consider the right hemisphere without loss of generality. By equation \eqref{eq:Fb}, we have $F_{b_{n-1}} > 0$ if
\begin{equation}
\tau > \frac{1}{(\gamma - 1)\, b_0 (1 - b_0)}.
\label{eq:taurange}
\end{equation}
By total differentiation,
\[
\frac{\partial S_n}{\partial b_{n-1}}
= -\frac{F_{b_{n-1}}}{F_S}.
\]
Recalling $F_S < 0$ from equation \eqref{eq:FS}, it follows that
\[
\frac{\partial S_n}{\partial b_{n-1}} > 0,
\]
whenever inequality \eqref{eq:taurange} holds.

\subsubsection{Proof of \Cref{prop:state0}}\label{state0proof} 
Consider the right hemisphere without loss of generality. Since
\[
\frac{\partial^2 b_{n}}{\partial b_{n-1} \partial b_{0}}
= \frac{\partial^2 S_n}{\partial b_{n-1} \partial b_{0}}
= \frac{\partial^2 S_n}{\partial b_{0} \partial b_{n-1}},
\]
it suffices to show that
\[
\frac{\partial^2 S_n}{\partial b_{0} \partial b_{n-1}} > 0,
\]
and that this cross-partial derivative increases with $n$.

A simple manipulation of equation \eqref{eq:FOC} yields
\begin{equation}
\frac{\partial b_n}{\partial b_{n-1}}
= \frac{\partial (b_{n-1} + S_n)}{\partial b_{n-1}}
= 1 + \frac{\partial S_n}{\partial b_{n-1}} > 1,
\label{eq:chainlogic}
\end{equation}
where $\frac{\partial S_n}{\partial b_{n-1}} > 0$ follows from \Cref{prop:size}. By equation \eqref{eq:chainlogic},
\begin{equation}
\frac{\partial b_n}{\partial b_{0}}
= \prod_{i=0}^{n-1} \frac{\partial b_{n-i}}{\partial b_{n-i-1}} > 0,
\label{eq:shock1}
\end{equation}
for any $n \ge 1$. Consequently,
\begin{equation}
\frac{\partial S_n}{\partial b_{0}}
= \frac{\partial S_n}{\partial b_{n-1}}
\frac{\partial b_{n-1}}{\partial b_{0}} > 0.
\label{eq:shock2}
\end{equation}
The first term in equation \eqref{eq:shock2} is positive and increasing in $n$. For larger $n$, $b_n$ must extend farther from $b_{n-1}$, which results in a larger $S_n$. 

We now turn to the second term in equation \eqref{eq:shock2}, which can be written as
\begin{equation*}
\frac{\partial b_{n-1}}{\partial b_0}
= \prod_{i=1}^{n-1} \frac{\partial b_{n-i}}{\partial b_{n-i-1}},
\end{equation*}
using equation \eqref{eq:shock1}. Each term in the product is weakly greater than one, and equals one if all states from $1$ to $n-1$ retain their original sizes. For larger $n$, the product contains an additional term and therefore weakly increases. 

Combining the two terms, we obtain $\frac{\partial^2 S_n}{\partial b_{0} \partial b_{n-1}}>0$ and that this cross-partial derivative increases with $n$. 

\subsubsection{Proof of \Cref{prop:SE}}\label{se}

Denote the ratio of marginal utility effects by
\begin{equation}
\theta_n \equiv 
\frac{\partial V(t)/\partial S_n}{\partial U(t)/\partial S_n}
=\frac{\psi\,C^l(t)^{-\gamma}}{C^z(t)^{-\gamma}}>0.
\label{eq:thetan}
\end{equation}
Hence, the border setter’s FOC becomes
\begin{equation}
(1+\phi\theta_n)\,\tau R_n^{\gamma-1}\big(1-b_{n-1}-S_n\big)=h.
\label{eq:SEFOC}
\end{equation}
By comparing equation \eqref{eq:SEFOC} with equation \eqref{eq:FOC}, $S_n^{\mathrm{SE}} > S_n^*$ is obtained for any $n$. Because the total length of the world is fixed, a weak expansion of every state implies that the total number of states must weakly fall: $N^{\mathrm{SE}} \le N^*$.

\subsubsection{Proof of \Cref{prop:consensus}}\label{consensusproof}
Consider the right hemisphere without loss of generality:
\begin{equation}
\frac{\partial \text{Var}(G_n)}{\partial \tau}
= \frac{1}{N} \sum_{n=1}^{N} G_n
\left(
\frac{\partial S_0}{\partial \tau}
+ 2\sum_{k=1}^{n-1}\frac{\partial S_k}{\partial \tau}
+ \frac{\partial S_{n}}{\partial \tau}
\right).
\end{equation}
Equation~\eqref{eq:FS} implies $F_\tau > 0$, while equation \eqref{eq:Ftau} implies $F_S < 0$.
Hence,
\[
\frac{\partial S_k}{\partial \tau}
= -\frac{F_\tau}{F_S} > 0,
\qquad \forall\, k \ge 1,
\]
which in turn implies $\frac{\partial \text{Var}(G_n)}{\partial \tau} > 0$.

Likewise,
\begin{equation}
\frac{\partial \text{Var}(G_n)}{\partial h}
= \frac{1}{N} \sum_{n=1}^{N} G_n
\left(
\frac{\partial S_0}{\partial h}
+ 2\sum_{k=1}^{n-1}\frac{\partial S_k}{\partial h}
+ \frac{\partial S_{n}}{\partial h}
\right).
\end{equation}
Equation~\eqref{eq:Fh} implies $F_h < 0$, while equation \eqref{eq:Ftau} implies $F_S < 0$.
Hence,
\[
\frac{\partial S_k}{\partial h}
= -\frac{F_h}{F_S} < 0,
\qquad \forall\, k \ge 1,
\]
which implies $\frac{\partial \text{Var}(G_n)}{\partial h} < 0$.

\subsection{Accommodating arbitrary geography with network analysis \label{sec:network}}

In this appendix, we model the world as a set of $N$ locales and represent each state as a subset of these locales. The following model allows for arbitrary geography of locales and illustrates the resulting strategic complications. The concepts and techniques used here draw on graph theory and the strategic network formation literature \citep{JW96,Jackson08}. Henceforth, locales are referred to as \emph{nodes}, which may establish \emph{links} with other nodes in order to form a state. A state corresponds to a \emph{component}: a subset of nodes connected (directly or indirectly) through links. For any linked pair of nodes $(i,j)$, one unit shipped from node $i$ to node $j$ delivers $0<\delta\le 1$ units, symmetrically in either direction. We have \(\delta = 1\) if and only if \(i\) and \(j\) are linked (i.e., belong to the same state). Following the model in the main text, every node supplies one unit of its own distinct good to every other node.\footnote{The Cobb-Douglas production and consumption structures in the main-text model imply that each good receives a fixed expenditure share, which is analogous to the ``one unit'' assumption here.} In this setting, the political map of the world is a \emph{network}, denoted by \(g\), consisting of a set of nodes and the links among them. Some nodes are linked, forming one or multiple states, while others remain isolated.\footnote{Formally, let \(V\) denote the set of nodes and \(G\) the set of all possible networks. Our analysis takes place on a graph represented by \(G=(V,g)\).}

Following the literature, utilities are defined at the level of potential links rather than at the level of nodes. The utility generated by a link between nodes \(i\) and \(j\) is
\begin{equation}
U_{ij} = (1 - \delta)\;-\;\eta G_{ij}\;-\;h N_{ij},
\end{equation}
where \(\eta G_{ij}\) represents the link cost (\(\eta>0\), \(G_{ij}\) is the bilateral distance between \(i\) and \(j\)), and \(h N_{ij}\) captures the congestion cost (\(h>0\), \(N_{ij}\) is the number of nodes in their component). Because transportation costs are symmetric, the value \(U_{ij}\) applies equally to node \(j\) and to node \(i\). Here, arbitrary geography refers to a general $\{G_{ij}\}$ matrix.

The relevant equilibrium concept is pairwise stability \citep{JW96}. A network is pairwise stable if:
\begin{enumerate}
    \item Within-component stability: $\forall$ existing links $(i,j)$:
    \[
    (1 - \delta) - \eta G_{ij} - h N_{ij}^{\text{current}} \geq 0,
    \]
    \item Between-component stability: $\forall$ non-existing links $(i,j)$:
    \[
    (1 - \delta) - \eta G_{ij} - h (N_i + N_j) \leq 0,
    \]
    where $N_i$ and $N_j$ are component sizes before linking.
\end{enumerate}
Unlike a Nash equilibrium, pairwise stability requires mutual consent for adding links, which makes it appropriate for modeling state formation. This constitutes the first structure we need for modeling state formation.\footnote{A separate literature, beginning with \citet{BG00}, adopts the Nash equilibrium by assuming that links are formed unilaterally. This approach does not fit our setting, because locales cannot unilaterally compel others to join a state.} 

Specifically, given a network partitioned into components $\{C_k\}$:
for all $(i,j)$ in the same component $C_k$ of size $n_k$:
\begin{equation}
(1 - \delta) - \eta G_{ij} - h n_k \geq 0, \quad \forall (i,j) \in C_k.
\end{equation}
Let $G_{\text{max}}^{\text{intra}} = \max_{(i,j)\in\text{same component}} G_{ij}$, then:
\begin{equation}
n_k \leq \frac{1 - \delta - \eta G_{\text{max}}^{\text{intra}}}{h}, \quad \forall k.
\end{equation}
And for all $(i,j)$ in different components $(C_k$, $C_l)$:
\begin{equation}
(1 - \delta) - \eta G_{ij} - h (n_k + n_l) \leq 0.
\end{equation}
Let $G_{\text{min}}^{\text{inter}} = \min_{(i,j)\in\text{different components}} G_{ij}$, then:
\begin{equation}
n_k + n_l \geq \frac{1 - \delta - \eta G_{\text{min}}^{\text{inter}}}{h}, \quad \forall k \neq l.
\end{equation}
Thus, a nontrivial \textbf{pairwise stable equilibrium} exists if there is a partition \(\{C_k\}\) such that:\footnote{A nontrivial pairwise stable equilibrium refers to a network that is not empty. In strategic network formation, the empty network---one with no links at all---is often pairwise stable.}
\begin{align}
&\max_{(i,j)\in\text{same component}} \left[\eta G_{ij} + h |C(i)| \right] \leq 1 - \delta \\
&\min_{(i,j)\in\text{different components}} \left[\eta G_{ij} + h (|C(i)| + |C(j)|)\right] \geq 1 - \delta
\end{align}
where $|C(i)|$ denotes the size of node $i$'s component (i.e., locale $i$'s state).

Several observations are in order. First, the equilibrium depends crucially on the underlying geography, namely the distance structure \(G_{ij}\). Second, a larger \(h\) leads to a smaller maximum component size, while a higher \(\eta\) requires shorter distances for links to be stable. Third, simple examples of such equilibria can be constructed numerically. For example, consider seven nodes located on a 2-D plane: \(A(0,0)\), \(B(1,0)\), \(C(0,1)\), \(D(1,1)\), \(E(4,4)\), \(F(5,4)\), and \(G(4,5)\). Let \(\delta = 0.2\), \(\eta = 0.1\), and \(h = 0.05\). Then nodes \(A\), \(B\), \(C\), and \(D\) form one component (state), while \(E\), \(F\), and \(G\) form another component (state). Thus, there are two states in this world. 

Importantly, because the equilibrium is not unique, the model offers limited insight into the precise political map of the modeled world, and even less into how its geopolitics operate. We now impose two additional structures to obtain a unique equilibrium:
\begin{description}
\item[{Activation order}] Nodes activate in order of increasing total connection cost \(TC_i = \sum_j G_{ij} \). When node \( i \) activates, it forms links with existing nodes if mutually beneficial. Assume that $\{TC_i\}$  has no tie.
\item[{Random error}] We introduce pair-specific random errors to the utility function:
\begin{equation}
\tilde{U}_{ij} = (1 - \delta) - \gamma G_{ij} - h N_{ij} + \epsilon_{ij},
\end{equation}
where $\epsilon_{ij}$ are i.i.d. random variables following a continuous distribution with zero mean and bounded support. Assume  $\epsilon_{ij} = \epsilon_{ji}$.
\end{description}
Taken together, when a node becomes active, it evaluates potential links with all other nodes according to the updated utility function and forms a link if and only if \(\tilde{U}_{ij} \ge 0\) (a condition that automatically holds for node \(j\) as well). Because \(G_{ij}\) enters negatively, each node prefers to form links with geographically closer nodes first.

The decision process described above is deterministic at the level of activation order. \(\{TC_i\}\), which contains no ties, determines a unique sequence in which nodes become active. The random errors then break all remaining ties in each activating node’s link decisions. For an activating node \(i\), if forming a link with node \(j\) yields positive utility, we consider whether node \(j\) has already activated. If node \(j\) activated before node \(i\), then \(j\) has already decided whether to link with \(i\), and the resulting \(ij\)-link is pairwise stable. If node \(j\) has not yet activated, then node \(i\) forms the link with \(j\). When node \(j\) activates later, it decides whether to sever the link with node \(i\). If node \(j\) chooses not to sever it, the \(ij\)-link is pairwise stable. The equilibrium selection is stochastic:
\begin{equation}
\mathbb{P}(\text{Network } G)
= \mathbb{P}(\epsilon \in \Omega_G),
\end{equation}
where \(\Omega_G\) is the set of shock realizations that generate network \(G\). Each \(\epsilon \in \Omega_G\) leads to a unique equilibrium.

The above example illustrates why the linear-world setting admits a unique equilibrium without the three additional structures. First, in the linear world, states consist of contiguous locales and all locales prefer to join their proximal neighbors. Consequently, a proximal locale’s unilateral selection of its neighbors is automatically pairwise stable. Second, linear geography naturally endows locales closer to the world’s geometric center (WGC) with greater power, eliminating the need to impose an exogenous activation order to rule out alternative equilibria. Third, because locations lie on a continuum along the world line, random utility shocks are unnecessary to break ties in locale selection.

Now it also becomes clear that the unique equilibrium (in the sense of pairwise stability) in the network-based world described above is difficult to use for characterizing geopolitical phenomena. The resulting political map must be characterized numerically using simulations. Because borders do not have analytical characterizations, their shapes and movements cannot be traced smoothly, limiting their usefulness for studying geopolitics.

\end{document}